\documentclass[12pt]{iopart}

\expandafter\let\csname equation*\endcsname\relax
\expandafter\let\csname endequation*\endcsname\relax
\usepackage{amsmath,amssymb}
\usepackage{graphicx}
\usepackage{booktabs}
\usepackage{bm}
\usepackage[hidelinks]{hyperref}

\newtheorem{theorem}{Theorem}
\newtheorem{proposition}{Proposition}
\newtheorem{lemma}{Lemma}
\newtheorem{corollary}{Corollary}
\newtheorem{definition}{Definition}
\newtheorem{remark}{Remark}
\newenvironment{proof}{\par\noindent\textit{Proof.}\ }{\hfill$\square$\par\medskip}

\newcommand{\C}{\mathcal{C}}
\newcommand{\V}{V}
\newcommand{\E}{E}
\newcommand{\Rtot}{R_{\mathrm{tot}}}
\newcommand{\Phid}{\Phi_{\mathrm{d}}}
\newcommand{\pd}{\partial}
\newcommand{\dd}{\mathrm{d}}
\newcommand{\ket}[1]{\lvert #1\rangle}
\newcommand{\A}{\mathcal{A}}
\DeclareMathOperator{\out}{out}
\DeclareMathOperator{\spec}{spec}

\begin{document}

\title[Time-resolved edge flux on the SU(3) triangle]{Time-Resolved Edge Flux on the
SU(3) Triangle: A Cooperative-Decay Graph Framework with Exact Bateman Solutions}

\author{G O Ariunbold$^{1,2,3,*}$ and V Sivaraman$^{4}$}

\address{$^1$ Department of Physics and Astronomy, Mississippi State University,
Mississippi State, MS 39762, USA}

\address{$^2$ Physics Department, Wesleyan University, Middletown, 06459, CT, USA}

\address{$^3$ Department of Physics and Astronomy, Texas A{\rm $\&$}M University, College Station, TX, USA}

\address{$^4$ Department of Mathematics and Statistics, Mississippi State University,
Mississippi State, MS 39762, USA}

\address{$^*$ Corresponding author}

\ead{ag2372@msstate.edu}

\begin{abstract}
Cooperative emission from three-level ensembles is conventionally diagnosed through the
radiated intensity alone, which cannot say which channels carry the photons at each instant.
We promote every transition of the cooperative-decay graph into a time-dependent edge flux and
show that node-wise conservation of this flux is an exact restatement of the Pauli master
equation. Suppose the graph carries a depth function that decreases by one along every edge.
Then the total cooperative rate is the descent speed of the mean graph depth, and its
energy-weighted form gives the radiated intensity as $I(t)=-\dd\langle E\rangle/\dd t$ for
arbitrary level spacings. The same ordering triangularizes the generator. This delivers the
relaxation spectrum from the state exit rates, together with a generalized Bateman closed form
for every population on an arbitrary finite directed acyclic graph, in which all exit-rate
degeneracies are absorbed into a polynomial-times-exponential recursion.
\end{abstract}

\vspace{2pc}

\vspace{1pc}


\maketitle

\section{Introduction}

In 1954 Dicke showed that an ensemble of identical two-level atoms confined to a
sub-wavelength volume cannot radiate independently: the vacuum field mediates a
collective coupling between them. The net macroscopic dipole moment of the atomic system is
zero, and yet the ensemble emits a coherent burst -- superradiance -- whose peak intensity
scales as $N^2$ rather than $N$~\cite{Dicke1954}. The effect was demonstrated experimentally
two decades later~\cite{Skribanowitz1973}. The
quantitative theory rests on the Pauli master equation of Bonifacio, Schwendimann,
and Haake~\cite{BSH1971}, in which the burst is diagnosed through a single aggregate
observable, the radiated intensity $I(t)$, obtained by solving the master equation and
summing the cooperative emission rates over all populated states. Four standard
benchmarks characterize the two-level burst (with $\Gamma$ the single-atom decay
rate): the peak intensity $I_{\max}\propto N^2$; the delay time
$\tau_D\simeq(E_0+\ln N)/(N\Gamma)$, with $E_0\approx0.5772$ the Euler--Mascheroni
constant; the temporal width $\tau_W\propto1/(N\Gamma)$; and the quantum fluctuation
of the delay time, $\sigma_D\simeq\pi/(\sqrt{6}\,N\Gamma)$~\cite{BSH1971,GrossHaroche1982}.
For the two-level ensemble this master equation was solved analytically early on:
Degiorgio and Ghielmetti gave a large-$N$ approximation~\cite{DegiorgioGhielmetti1971}, and
Lee obtained the exact solution -- for complete~\cite{LeeI} and then arbitrary~\cite{LeeII}
initial excitation -- as a time-evolution matrix whose entries are
polynomial-times-exponential in $t$. The two-level problem has attracted renewed analytic
attention. Holzinger and Genes~\cite{HolzingerGenes2025} recast the solution as a finite sum of
residues of a rational generating function whose poles are the Dicke ladder rates; the
degeneracies below the equator are supplied there by second-order poles. A recent compendium
collects the rate-equation, non-Hermitian, combinatorial and quantum-jump routes to the same
answer~\cite{Holzinger2026}. That Bateman-type structure, together with the secular terms forced
by palindromic pole coincidences, is what the present framework carries onto the multi-level
state graph. The same repeated-root mechanism is shown to act on graphs whose interior nodes
have several parents -- where no generating function in a single variable is available.

The master-equation theory extends to multi-level
atoms~\cite{Agarwal1970,Agarwal1973,CKG1975}. For $N$ identical three-level atoms in the
small-sample (Dicke) limit, cooperative decay comes in exactly three dipole-allowed
topologies. In the Ladder (cascade) the atom descends sequentially $1\to2\to3$. In the
$\Lambda$ a common upper level feeds two lower levels. In the Vee two upper levels feed a
common ground level. Their exact numerical solutions have recently been obtained and
studied as a trilogy: the cascade superradiance model~\cite{Ariunbold2022}, the
mode-selective $\Lambda$~\cite{ABmodeselective2025}, and the superradiant synthesis of
the Vee~\cite{ABsynthesis2026}, whose two initially independent sub-ensembles were
observed to merge dynamically into a single macroscopic collective state. The same
master-equation programme has since been carried to four levels: in the fully symmetric SU(4)
representation the occupation-number basis becomes a tetrahedral weight lattice on which six
embedded $\mathfrak{su}(2)$ transition subalgebras act, and all seven dipole-allowed four-level
topologies are reduced to one Pauli-type rate equation~\cite{Lutsukh2026}. The lattice
formulation used below is the SU(3) member of that family, and the graph-theoretic results of
Section~\ref{sec:structure} are stated so that they apply verbatim to the SU(4) tetrahedron and
beyond.

In every one of these treatments, the observable used to characterize the burst is
the radiated intensity $I(t)$: a single scalar curve obtained by solving the
population master equation and summing the $\hbar\omega$-weighted cooperative rate over
every populated state and every emission channel. This aggregate answers \emph{whether}
and \emph{when} a burst occurs, but not \emph{how} it occurs at the level of the
underlying state graph: it cannot say whether, at a given instant, the decaying
population is funneling through one dominant sequence of transitions or spreading
simultaneously across many.

The present paper supplies that missing, edge-resolved description, and does so for
all three three-level topologies at once. Consider the cooperative-decay graph, whose nodes are the joint atomic
population states and whose directed edges are the allowed emission channels, each carrying a
state-dependent, bosonic-enhanced rate. We treat it not as a static object used only to write
down the master equation, but as the support of a literal time-dependent flow. Every edge
carries an instantaneous flux, equal to its rate times the population of its source state.
Conservation of this flux at every node is then shown to be nothing but the master equation
itself, written locally rather than globally. From this starting point we obtain five results. The radiated
intensity is recovered exactly as the energy-weighted sum of all edge fluxes. A depth-descent
identity expresses the total cooperative rate as minus the rate of descent of the population's
mean graph depth, for any uniformly leveled decay graph. The full relaxation spectrum follows
from the state exit rates alone. A generalized Bateman closed form gives every state population
exactly, accommodating arbitrary exit-rate coincidences. Finally, the complete decay pathways
are counted in closed form.

The paper is organized as follows. Section~\ref{sec:framework} introduces the general
time-resolved edge-flux (TREF) framework on the SU(3) triangular lattice, establishes the
conservation, depth-descent, and energy identities, proves that all three three-level
topologies are uniformly leveled, and records their structural fingerprints
(Table~\ref{tab:fingerprints}). Section~\ref{sec:structure} develops the algebraic and
combinatorial consequences of the depth-leveled acyclic structure. These are the
triangularization of the generator, the relaxation spectrum from the diagonal, and the general
Bateman closed form for a finite directed acyclic graph (DAG), together with the polynomial-exponential machinery that removes
any degeneracy restriction. It closes with the pathway counts and the palindrome-forced
spectral degeneracies. Section~\ref{sec:recursion} applies the general recursion to all three
configurations. Section~\ref{sec:conclusions} concludes.

\section{Time-resolved edge flux on the SU(3) triangular lattice}
\label{sec:framework}

\subsection{States, topologies, and the master equation}
\label{sec:states}

Let $N$ identical three-level atoms have internal levels $\ket1,\ket2,\ket3$ with
energies $E_1>E_2>E_3$, all coupled to common field modes (small-sample Dicke limit).
The permutation-symmetric subspace carries the fully symmetric SU(3) representation,
spanned by occupation states $\ket{\bm q}=\ket{q_1,q_2,q_3}$ with $q_1+q_2+q_3=N$,
forming a triangular lattice with
\begin{equation}
D_N=\frac{(N+1)(N+2)}{2}
\end{equation}
sites. The three topologies are the three admissible channel sets
\begin{equation}
\text{Ladder: }\{1{\to}2,\,2{\to}3\};\qquad
\Lambda: \{1{\to}2,\,1{\to}3\};\qquad
\text{Vee: }\{1{\to}3,\,2{\to}3\}.
\label{eq:channelsets}
\end{equation}
In the Born--Markov approximation the coherences decouple from the populations. The
occupation probability $p(\bm q,t)$ then obeys the Pauli master
equation~\cite{Agarwal1973,Ariunbold2022,Lindblad1976,BreuerPetruccione2002}
\begin{equation}
\frac{\pd p(\bm q,t)}{\pd t}
=\sum_{(m\to n)\in\C}\Big[\,\Gamma_{nm}\,q_n(q_m{+}1)\,p(\bm q-\bm e_n+\bm e_m,t)
-\Gamma_{nm}\,q_m(q_n{+}1)\,p(\bm q,t)\,\Big],
\label{eq:master}
\end{equation}
with channel set $\C$, single-atom rates $\Gamma_{nm}$, unit vector $\bm e_i$ in the
$q_i$ direction, and edge rates
\begin{equation}
w_{nm}(\bm q)=\Gamma_{nm}\,q_m(q_n+1)\qquad(q_m\ge1),
\label{eq:edgerate}
\end{equation}
the product of the emitting population $q_m$ and the bosonic enhancement $(q_n+1)$ of
the receiving level. Each channel contributes $N(N+1)/2$ directed edges. For the
Ladder, Eq.~\eqref{eq:master} is the cascade master equation of~\cite{Ariunbold2022} in
occupation coordinates ($q_1=n$, $q_2=m-n$, $q_3=N-m$ in the two-index notation used
there); for the $\Lambda$ and Vee it is the master equation of~\cite{ABmodeselective2025}
and~\cite{ABsynthesis2026}, respectively. We write $\V$ for the set of occupation
states and $\E$ for the set of directed edges $(\bm u\to\bm v)$ with positive rate.

\subsection{The depth function and the edge flux}
\label{sec:depth}

Two concepts underlie the entire framework: a graph-theoretic device that certifies the
decay graph is acyclic and drives the identity of Section~\ref{sec:descent}, and the
central physical quantity of the paper -- the edge flux.

\begin{definition}[Depth function]\label{def:depth}
A depth function on $(\V,\E)$ is a map $\Phid:\V\to\mathbb{Z}_{\ge0}$ that strictly
decreases along every directed edge, $\Phid(\bm v)<\Phid(\bm u)$ whenever
$(\bm u\to\bm v)\in\E$. A decay graph admitting such a function is necessarily acyclic.
The graph is \emph{uniformly leveled} by $\Phid$ if every edge decreases $\Phid$ by
exactly~$1$.
\end{definition}

\begin{definition}[Edge flux]\label{def:flux}
For every directed edge of the decay graph, the instantaneous flux through it is
\begin{equation}
\Phi_{nm}(\bm q;t):=w_{nm}(\bm q)\,p(\bm q,t)\;\ge0,
\label{eq:fluxdef}
\end{equation}
wherever the edge exists, and $0$ otherwise. The decay graph $(\V,\E)$ with rates
$\{w_{nm}\}$, together with the time-dependent flux family
$\{\Phi_{nm}(\cdot;t)\}_{t\ge0}$, constitutes the time-resolved edge-flux network.
\end{definition}

The edge flux turns the cooperative-decay graph from a static wiring diagram into a
literal flow network, one that changes with time as the atomic population redistributes
itself across the state space.

\subsection{Exact conservation and recovery of the total cooperative rate}
\label{sec:conservation}

The first thing to establish is that nothing has been added by this change of viewpoint. The
edge flux is built from quantities the master equation already contains, and the following
proposition says that reassembling them node by node returns the master equation exactly --
the flux picture is a rewriting, not a model.

\begin{proposition}[Exact flux conservation]\label{prop:conservation}
For every topology, every lattice state $\bm q\in\V$, and every $t$,
\begin{equation}
\frac{\dd p(\bm q,t)}{\dd t}
=\underbrace{\sum_{(m\to n)\in\C}\Phi_{nm}(\bm q+\bm e_m-\bm e_n;t)}_{\text{flux in}}
-\underbrace{\sum_{(m\to n)\in\C}\Phi_{nm}(\bm q;t)}_{\text{flux out}}.
\label{eq:conservation}
\end{equation}
The master equation~\eqref{eq:master} is exactly this local conservation law, stated at
every node individually.
\end{proposition}

\begin{proof}
Substitute Definition~\ref{def:flux} into~\eqref{eq:master}; the two sides are identical
by construction.
\end{proof}

This local form is strictly stronger than global probability conservation. The latter
follows by summing~\eqref{eq:conservation} over all nodes and telescoping. The node-by-node
statement, by contrast, records \emph{which} transitions are responsible for each state's
population change at each instant -- information invisible in any aggregate scalar
description.

Summing the fluxes instead of differencing them produces the observables. Two are useful, and
it is worth keeping them apart: one counts transitions, the other counts energy.

\begin{proposition}[Channel-resolved and total cooperative rates]\label{prop:rates}
For each open channel let $R_{nm}(t):=\sum_{\bm q}\Phi_{nm}(\bm q;t)$ be its integrated
flux. The total cooperative (transition-count) rate is
$\Rtot(t)=\sum_{(m\to n)\in\C}R_{nm}(t)$, and the
physical radiated intensity is the photon-energy-weighted flux sum
\begin{equation}
I(t)=\sum_{(m\to n)\in\C}\hbar\omega_{nm}\,R_{nm}(t),\qquad
\hbar\omega_{nm}=E_m-E_n .
\label{eq:intensity}
\end{equation}
\end{proposition}

\begin{remark}\label{rem:count}
$\Rtot(t)$ counts transitions, not radiated energy: it weights channels with different
photon energies identically, and reduces to a constant multiple of $I(t)$ only in the
equally spaced case. We retain it as a separate observable because it drives the purely
combinatorial depth-descent identity below.
\end{remark}

\subsection{The depth-descent identity}
\label{sec:descent}

So far the flux has only been summed. The graph structure enters when the sum is weighted by
depth: if every edge costs exactly one unit of depth, then each emission moves the population
exactly one level down, and the total emission rate must equal the speed at which the
population descends. That intuition is exact.

\begin{theorem}[Depth-descent identity]\label{thm:descent}
Let the decay graph be uniformly leveled by a depth function $\Phid$
(Definition~\ref{def:depth}), and let $\langle\Phid\rangle(t):=\sum_{\bm q}\Phid(\bm q)\,p(\bm q,t)$
be the population's mean depth. Then
\begin{equation}
\Rtot(t)=-\frac{\dd}{\dd t}\langle\Phid\rangle(t).
\label{eq:descent}
\end{equation}
\end{theorem}

\begin{proof}
Differentiate $\langle\Phid\rangle$ and substitute the gain--loss form
of~\eqref{eq:conservation}:
$\dd\langle\Phid\rangle/\dd t=\sum_{\bm u}\Phid(\bm u)\dot p_{\bm u}
=\sum_{\bm u}\Phid(\bm u)\sum_i[\Phi_i(\mathrm{pred}_i(\bm u);t)-\Phi_i(\bm u;t)]$.
Reindex each gain sum by $\bm u'=\mathrm{pred}_i(\bm u)$; uniform leveling gives
$\Phid(\bm u)=\Phid(\bm u')-1$ along every edge, so the $\Phid$-weighted gain terms
cancel the loss terms up to the extra $-1$, leaving
$\dd\langle\Phid\rangle/\dd t=-\sum_i\sum_{\bm u'}\Phi_i(\bm u';t)=-\Rtot(t)$.
\end{proof}

The identity holds precisely because the depth function decreases by exactly one along
every edge, regardless of channel: it is a structural consequence of uniform leveling,
not of the particular cooperative rate functions.

\subsection{Energy-weighted depth: the exact intensity}
\label{sec:energy}

Replacing the combinatorial depth by a channel-weighted one recovers the physical
intensity exactly.

\begin{theorem}[Energy-weighted descent identity]\label{thm:energy}
Let $\Psi(\bm q):=\sum_i E_i q_i$ be the total excitation energy of state $\bm q$
(relative to the ground level, up to a constant). Every channel-$(m\to n)$ edge
decreases $\Psi$ by exactly $\hbar\omega_{nm}=E_m-E_n$, so the radiated
intensity~\eqref{eq:intensity} satisfies
\begin{equation}
I(t)=-\frac{\dd}{\dd t}\langle\Psi\rangle(t)=-\frac{\dd}{\dd t}\langle E\rangle(t),
\label{eq:energyid}
\end{equation}
for any level spacing.
\end{theorem}

\begin{proof}
Identical to the proof of Theorem~\ref{thm:descent}, with each channel-$(m\to n)$ edge
decreasing the weighting function by the channel-dependent constant $\hbar\omega_{nm}$
rather than by $1$.
\end{proof}

Equation~\eqref{eq:energyid} is the exact energy-conservation reading of the radiated
intensity: $-\dd\langle E\rangle/\dd t$ is the rate at which the mean stored atomic
energy decreases, i.e.\ the power radiated into the field. A complementary,
state-resolved cross-check confirms that the formalism reproduces the established
microscopic rate expressions. For every state, the channel-weighted exit rate
$\sum_{(m\to n)}\hbar\omega_{nm}w_{nm}(\bm q)$ agrees term by term with Agarwal's
per-state cascade radiation rate~\cite{Agarwal1973}. Hence
$I(t)=\sum_{\bm q}p(\bm q,t)\sum_{(m\to n)}\hbar\omega_{nm}w_{nm}(\bm q)$ is the exact
ensemble average of Agarwal's expression at every fixed $t$. This is an algebraic identity
requiring no time derivative, and it is distinct from -- and more elementary than -- the
dynamical statement~\eqref{eq:energyid}.

\subsection{All three topologies are uniformly leveled}
\label{sec:leveling}

Both identities were stated conditionally: they hold for graphs that are uniformly leveled.
That condition now has to be checked, and for three-level systems it costs nothing --- every
topology satisfies it, and the depth function can be written down by inspection.

A linear depth function $\Phi_{\bm c}(\bm q)=\sum_i c_i q_i$ decreases by exactly one
along every edge iff $c_m-c_n=1$ for every open channel. With only three levels no
channel set can contain two directed paths of unequal length between the same pair of
levels, so every three-level topology is graded:
\begin{equation}
\bm c_{\text{Ladder}}=(2,1,0),\qquad
\bm c_{\Lambda}=(1,0,0),\qquad
\bm c_{\text{Vee}}=(1,1,0),
\label{eq:weights}
\end{equation}
and the depth-descent identity~\eqref{eq:descent} and the energy
identity~\eqref{eq:energyid} hold verbatim for each. Physically, $\Phi_{\bm c}$ counts
the number of photons the system has yet to emit: $2q_1+q_2$ for the Ladder (each
top-level atom owes two photons), $q_1$ for the $\Lambda$, and $q_1+q_2$ for the Vee
(each excited atom owes one).

\subsection{Structural fingerprints}
\label{sec:fingerprints}

Although the dynamical identities are common, the three graphs differ sharply in
structure; Table~\ref{tab:fingerprints} collects the results, all special cases of the
general theorems of Section~\ref{sec:structure} evaluated on the triangle.

\emph{Stationary manifold.} By triangularity of the generator under depth ordering
(Lemma~\ref{lem:triangular}), the multiplicity of the zero eigenvalue equals the number
of absorbing states -- states in which every emitting level is empty. The Ladder and Vee
funnel into the single vertex $\ket{0,0,N}$; the $\Lambda$ terminates on the entire edge
$q_1=0$, an $(N{+}1)$-fold degenerate dark manifold whose final distribution encodes the
branching history.

\emph{Palindrome degeneracies.} On a boundary edge supporting exactly one active channel
$a\to b$ with level $b$ dark, the exit rate $r=\Gamma_{ba}q_a(N-q_a+1)$ is palindromic
under $q_a\mapsto N+1-q_a$, forcing $\lfloor N/2\rfloor$ exact eigenvalue pairs for all
rate values (Section~\ref{sec:palindrome}). The criterion selects one family for the
Ladder, two for the Vee, and none for the $\Lambda$.

\emph{Pathway counts.} Complete decay pathways from the physical initial states are the
Catalan number $C_N$ for the Ladder (Dyck-path bijection), $2^N$ for the $\Lambda$
($N$ unconstrained binary branch choices), and $\binom{N}{a}$ for the Vee (interleavings
of the two reservoirs' departures from $\ket{a,b,0}$).

\begin{table}[t]
\centering
\caption{The three dipole-allowed three-level topologies as TREF decay graphs on the
SU(3) triangular lattice ($D_N=(N{+}1)(N{+}2)/2$ states; $N(N{+}1)/2$ edges per channel).
All three are uniformly leveled. Initial states: $\ket{N,0,0}$ (Ladder, $\Lambda$),
$\ket{a,b,0}$ with $a+b=N$ (Vee). $C_N$: Catalan number.}
\label{tab:fingerprints}
\setlength{\tabcolsep}{5pt}
\resizebox{\textwidth}{!}{%
\begin{tabular}{@{}llccccc@{}}
\toprule
Topology & Channels & Depth $\bm c$ & $\dim\ker Q$ & Absorbing set & Palindr.\ families & Pathways\\
\midrule
Ladder   & $1{\to}2,\,2{\to}3$ & $(2,1,0)$ & $1$    & vertex $\ket{0,0,N}$ & $1$ & $C_N$\\
$\Lambda$& $1{\to}2,\,1{\to}3$ & $(1,0,0)$ & $N{+}1$ & edge $q_1{=}0$      & $0$ & $2^N$\\
Vee      & $1{\to}3,\,2{\to}3$ & $(1,1,0)$ & $1$    & vertex $\ket{0,0,N}$ & $2$ & $\binom{N}{a}$\\
\bottomrule
\end{tabular}}
\end{table}

\section{Structural and combinatorial properties of the decay graphs}
\label{sec:structure}

The depth function of Definition~\ref{def:depth} was introduced to certify acyclicity
and to drive the depth-descent identity. The same depth-leveled, acyclic structure has
further consequences that are independent of the flux dynamics: it controls the
algebraic structure of the generator itself and the combinatorics of complete decay
pathways. We develop these here in the general finite-DAG setting, then specialize to
the triangle.

\subsection{Depth ordering triangularizes the generator}
\label{sec:triangular}

\begin{definition}[Generator]\label{def:generator}
For $\bm u\in\V$, let $r(\bm u):=\sum_{(m\to n)}w_{nm}(\bm u)$ be the total exit rate.
Collecting the populations into a vector $P(t)$ indexed by $\V$, the master
equation~\eqref{eq:master} is $\dot P(t)=QP(t)$, where $Q_{\bm v\bm u}=w_{nm}(\bm u)$ if
$\bm u\to\bm v$ is the channel-$(m\to n)$ edge, $Q_{\bm v\bm u}=0$ otherwise, and
$Q_{\bm u\bm u}=-r(\bm u)$.
\end{definition}

Written this way the generator looks like a general sparse matrix. It is not: the depth
function orders the states, and in that order no edge can ever point backwards.

\begin{lemma}[Triangularization]\label{lem:triangular}
List the states of $\V$ in order of strictly decreasing depth $\Phid$ (ties broken
arbitrarily). In this ordering $Q$ is lower triangular: every nonzero off-diagonal
entry $Q_{\bm v\bm u}$ lies strictly below the diagonal.
\end{lemma}

\begin{proof}
$Q_{\bm v\bm u}\ne0$ for $\bm v\ne\bm u$ requires $\bm u\to\bm v\in\E$, hence
$\Phid(\bm v)<\Phid(\bm u)$; in decreasing-depth order $\bm v$ is listed strictly after
$\bm u$, so its row index exceeds $\bm u$'s column index.
\end{proof}

Triangularity is worth more than it looks. A triangular matrix wears its spectrum on its
diagonal, so the relaxation rates of the whole ensemble can be read off state by state,
without ever forming --- let alone diagonalizing --- the generator.

\begin{theorem}[Spectrum from the diagonal]\label{thm:spectrum}
The eigenvalues of $Q$ are exactly $\{-r(\bm u):\bm u\in\V\}$, obtainable in
$O(|\V|+|\E|)$ arithmetic operations rather than by $O(|\V|^3)$ diagonalization.
\end{theorem}

\begin{proof}
For a triangular matrix $\det(Q-\lambda I)=\prod_{\bm u}(Q_{\bm u\bm u}-\lambda)$, by
Lemma~\ref{lem:triangular}.
\end{proof}

\subsection{An exact, Bateman-type closed-form solution on any finite DAG}
\label{sec:bateman}

Triangularity turns the master equation into a lower-triangular ordinary differential equation (ODE) system in which
every state's population depends only on the already-solved populations of its
predecessors. This is a purely graph-theoretic fact.

\begin{theorem}[Bateman solution on any finite DAG]\label{thm:batemanDAG}
Let $G=(\V,\E)$ be a finite DAG with a depth function $\Phid$ satisfying
$\Phid(\bm v)<\Phid(\bm u)$ for every edge, exit rate $r(\bm u)=\sum_{\bm v}w_{\bm u\bm v}>0$
for each non-sink $\bm u$, edge weight $w_{\bm u\bm v}>0$, and $P_{\bm s}(0)=1$ at a
unique source $\bm s$. Then:
\begin{itemize}
\item[(i)] Any decreasing-depth ordering renders $Q$ lower triangular, with diagonal
$Q_{\bm u\bm u}=-r(\bm u)$.
\item[(ii)] $\spec(Q)=\{-r(\bm u):\bm u\in\V\}$.
\item[(iii)] $P_{\bm s}(t)=e^{-r(\bm s)t}$, and for every other $\bm u$,
\begin{equation}
P_{\bm u}(t)=\sum_{\bm p:\,\bm p\to\bm u}w_{\bm p\bm u}\int_0^t e^{-r(\bm u)(t-s)}P_{\bm p}(s)\,\dd s,
\label{eq:varofconst}
\end{equation}
solvable exactly by forward substitution in depth order. The solution has the
polynomial-exponential form $P_{\bm u}(t)=\sum_{\bm w\in\A(\bm u)}q_{\bm u,\bm w}(t)\,e^{-r(\bm w)t}$,
where $\A(\bm u)$ is the ancestor set and each $q_{\bm u,\bm w}$ is a polynomial in $t$ of
degree at most one less than the multiplicity of $r(\bm w)$ among the ancestor exit rates.
\end{itemize}
\end{theorem}

The classical 1910 result of Bateman~\cite{Bateman1910} is the special case in which the
graph is a simple chain. For the two-level Dicke chain the corresponding
exact solution -- including the degenerate case treated in Section~\ref{sec:degeneracy}
below -- was obtained by Lee~\cite{LeeI,LeeII}. His time-evolution matrix is precisely the
chain propagator of Theorem~\ref{thm:batemanDAG}(iii), and the double poles of his Laplace
transform are exactly the palindromic coincidences of Section~\ref{sec:palindrome}. The
novelty here is therefore not the chain solution, but its transfer to graphs whose interior
states have several parents. On the SU(3) triangle each interior state has exactly two parents,
one per channel. The ancestor set consequently grows as $O(N^2)$, and the recursion costs
$O(|\A(\bm u)|)$ per state.

\subsection{Removing the degeneracy restriction}
\label{sec:degeneracy}

The non-degeneracy condition (all ancestor exit rates distinct) fails exactly where two
ancestors in the same chain share the same exit rate; the naive coefficient
$1/(r(\bm u)-r(\bm w))$ then diverges, and the correct response is to raise the
polynomial degree. The key technical ingredient is a single closed-form evaluation
covering both cases.

\begin{lemma}[Polynomial-exponential integral]\label{lem:Jk}
For $a\in\mathbb{R}$, $k\in\mathbb{N}_0$, and $t\ge0$,
\begin{equation}
J_k(a,t):=\int_0^t s^k e^{as}\,\dd s=
\begin{cases}
\displaystyle e^{at}\sum_{j=0}^{k}(-1)^j\frac{k!}{(k-j)!}\frac{t^{k-j}}{a^{j+1}}
+(-1)^{k+1}\frac{k!}{a^{k+1}}, & a\ne0,\\[2.2ex]
\displaystyle \frac{t^{k+1}}{k+1}, & a=0.
\end{cases}
\label{eq:Jk}
\end{equation}
\end{lemma}

With $J_k$ in hand the forward substitution can be carried out symbolically, and the two
branches of~\eqref{eq:Jk} become the two branches of the recursion: a coincidence is not a
special case to be avoided but simply the $a=0$ line, which promotes $t^k$ to $t^{k+1}$.

\begin{theorem}[General Bateman closed form, no degeneracy restriction]\label{thm:batemangen}
Set $P_{\bm s}(t)=e^{-r(\bm s)t}$ at the source. For every other state $\bm u$, with
incoming channel edges from parents $\bm p$ whose populations are already known as
$P_{\bm p}(t)=\sum_{\bm w\in\A(\bm p)}Q_{\bm p,\bm w}(t)e^{-r(\bm w)t}$,
$Q_{\bm p,\bm w}(t)=\sum_{k=0}^{d_{\bm p,\bm w}}c_{\bm p,\bm w,k}t^k$, the population of
$\bm u$ is $P_{\bm u}(t)=\sum_{\bm w\in\A(\bm u)}Q_{\bm u,\bm w}(t)e^{-r(\bm w)t}$, with
coefficients obtained from every parent $\bm p$ and every term $(\bm w,k)$ by:
if $r(\bm w)\ne r(\bm u)$ (set $a=r(\bm u)-r(\bm w)$),
\begin{align}
c_{\bm u,\bm w,l}&\mathrel{+}=w_{\bm p\bm u}\,c_{\bm p,\bm w,k}\,(-1)^{k-l}\frac{k!}{l!\,a^{\,k-l+1}},
\quad l=0,\dots,k,\\
c_{\bm u,\bm u,0}&\mathrel{+}=w_{\bm p\bm u}\,c_{\bm p,\bm w,k}\,(-1)^{k+1}\frac{k!}{a^{\,k+1}};
\end{align}
if $r(\bm w)=r(\bm u)$ exactly,
\begin{equation}
c_{\bm u,\bm u,k+1}\mathrel{+}=\frac{w_{\bm p\bm u}\,c_{\bm p,\bm w,k}}{k+1},
\end{equation}
so that every exact rate coincidence encountered along a chain raises the polynomial
degree of that exponential by exactly one. This recursion is exact for every state of
every such graph, regardless of the coincidence pattern among $\{r(\bm w):\bm w\in\A(\bm u)\}$.
\end{theorem}

\begin{proof}
Substitute the known form of $P_{\bm p}(s)$ into~\eqref{eq:varofconst}; each term reduces
to $c_{\bm p,\bm w,k}e^{-r(\bm u)t}J_k(a,t)$ with $a=r(\bm u)-r(\bm w)$ and $J_k$ as in
Lemma~\ref{lem:Jk}. The case $a\ne0$ gives the first two lines by expanding $e^{at}$; the
case $a=0$ gives the third directly. Induction on depth completes the recursion.
\end{proof}

Theorem~\ref{thm:batemangen} does not need to know \emph{why} a coincidence occurs, only
that the rate gap is exactly zero; it accommodates any number of coincidences along a
single chain. Along a one-dimensional chain this reproduces the secular
($t\,e^{-\lambda t}$) terms that Lee~\cite{LeeI,LeeII} obtained for the two-level Dicke
master equation from the double poles of its Laplace transform; the statement here is the
multi-parent-DAG generalization. A practical caution: the coincidence
test must be tied to the working precision, since a fixed relative tolerance misclassifies
genuinely distinct but closely spaced rates as exact coincidences -- a discrete branching
error, not a roundoff one, diagnosed by its failure to shrink under additional precision.

\subsection{Counting complete decay pathways}
\label{sec:paths}

The generator governs how fast the population moves; the graph also fixes how many routes it
has to choose among. This count is purely combinatorial --- it depends on the channel set and
the initial state, not on the rates --- and it separates the three topologies as sharply as
anything dynamical does.

\begin{theorem}[Pathway counts]\label{thm:paths}
The number of complete directed paths from the physical initial state to the absorbing
set is $C_N=\binom{2N}{N}/(N+1)$ for the Ladder (from $\ket{N,0,0}$), $2^N$ for the
$\Lambda$ (from $\ket{N,0,0}$), and $\binom{N}{a}$ for the Vee (from $\ket{a,b,0}$).
\end{theorem}

\begin{proof}
\emph{Ladder}: encode a path as a length-$2N$ word over $\{1,2\}$ recording the channel
used; the lattice constraint that $\ket q$ always has $q_1\le q_1+q_2$ translates into the
ballot condition (every prefix has at least as many $1$'s as $2$'s), a Dyck word, counted
by $C_N$~\cite{Stanley2015}. \emph{$\Lambda$}: each of the $N$ departures from level $1$
chooses target $2$ or $3$ freely, giving $2^N$ words, of which $\binom{N}{k}$ end at the
dark state $\ket{0,k,N-k}$. \emph{Vee}: a path is an interleaving of the $a$ level-$1$ and
$b$ level-$2$ departures, $\binom{N}{a}$ orders.
\end{proof}

The Ladder count grows exponentially, $C_N\sim4^N/(N^{3/2}\sqrt\pi)$, while the edge count
$|\E|$ grows only quadratically: the burst's funneling into a compact set of states is a
non-trivial feature of the cooperative rate structure, not a consequence of there being
few routes.

\subsection{Palindrome-forced spectral degeneracies}
\label{sec:palindrome}

Theorem~\ref{thm:spectrum} makes the spectrum easy to compute but says nothing about whether
its entries are distinct, and Theorem~\ref{thm:batemangen} has just shown that coincidences
are what generate secular terms. It is therefore worth asking where coincidences are forced by
the geometry rather than arriving by accident. The answer, on the boundary, is a reflection
symmetry of the exit rate.

\begin{lemma}[Row palindrome]\label{lem:palindrome}
Let $L_{ab}$ be a boundary edge, the states with $q_a+q_b=N$ and all other $q_i=0$. If it
supports exactly one active channel $a\to b$ with level $b$ dark, the total exit rate
reduces to $r(q_a)=\Gamma_{ba}\,q_a(N-q_a+1)$, palindromic under $q_a\mapsto N+1-q_a$.
\end{lemma}

\begin{corollary}[Forced spectral degeneracy]\label{cor:degeneracy}
By Theorem~\ref{thm:spectrum} the palindrome of Lemma~\ref{lem:palindrome} forces
$\lfloor N/2\rfloor$ exact eigenvalue pairs of $Q$, for every $N$ and every rate value.
The criterion ($\out(a)=\{b\}$, $\out(b)=\varnothing$) selects one such family for the
Ladder (its row $q_1=0$), two for the Vee (both edges $L_{13}$ and $L_{23}$), and none
for the $\Lambda$ (whose absorbing edge $q_1=0$ carries two active channels).
\end{corollary}

On the single-channel boundary edge this palindrome is the $s_i=s_{N-i-1}$ pole coincidence
identified by Lee~\cite{LeeI,LeeII} for the two-level chain; the content of
Corollary~\ref{cor:degeneracy} is the graph-theoretic criterion that says which boundary
edges of which topology inherit it. The boundary palindromes are not, however, the only
rate-independent coincidences. Away from the boundary two exit rates can coincide \emph{term
by term} -- each channel's edge factor matching separately -- so that the tie holds for
\emph{every} value of the rates, rational or not. For the $\Lambda$
[$r=q_1(\Gamma_{21}(q_2{+}1)+\Gamma_{31}(q_3{+}1))$] such ties pair states across the
inter-row reflection $q_1\mapsto N{+}2{-}q_1$; the smallest case,
$r(N,0,0)=r(2,\tfrac N2{-}1,\tfrac N2{-}1)$ for even $N\ge4$, involves the source itself, so
generic-$\Lambda$ populations below the partner state carry degree-one secular terms at
\emph{all} rates. For the Vee the analogous ties pair emission shells across the Dicke
reflection $q_3\mapsto N{-}1{-}q_3$ (at $N=14$: five $\Lambda$ ties; $8$ interior Vee ties in
addition to the $2\lfloor N/2\rfloor=14$ boundary ones). These reflection-forced interior
ties, together with any further accidental coincidences at special rational rate ratios, are
handled without exclusion by the generalized recursion of Theorem~\ref{thm:batemangen}; their
systematic classification is developed in a companion note.

\section{The general recursion on all three topologies}
\label{sec:recursion}

Every one of the three graphs is acyclic and depth-ordered. The generator is therefore lower
triangular and its spectrum is $\{-r(\bm q)\}$, so Theorem~\ref{thm:batemangen} applies:
every population is a finite sum
\begin{equation}
p(\bm q,t)=\sum_{\bm w\in\A(\bm q)}Q_{\bm q,\bm w}(t)\,e^{-r(\bm w)t},
\label{eq:polyexp}
\end{equation}
over the ancestor set, with polynomial coefficients computed by forward substitution in
depth order. On the triangle each interior state has exactly two parents:
$(q_1{+}1,q_2{-}1,q_3)$ and $(q_1,q_2{+}1,q_3{-}1)$ for the Ladder;
$(q_1{+}1,q_2{-}1,q_3)$ and $(q_1{+}1,q_2,q_3{-}1)$ for the $\Lambda$;
$(q_1{+}1,q_2,q_3{-}1)$ and $(q_1,q_2{+}1,q_3{-}1)$ for the Vee.

For the Ladder the forced palindrome ties on the row $q_1=0$ (Lemma~\ref{lem:palindrome})
produce degree-one polynomial terms. For the $\Lambda$ and Vee most ancestor
chains are coincidence-free (pure exponential sums), but even generic rates meet the
isolated reflection-forced ties of Section~\ref{sec:palindrome} along some chains, each
raising a polynomial degree by one; further accidental coincidences appear at special
rational rate ratios. The recursion is indifferent to all of these: it detects ties by
exact rate comparison and absorbs them into polynomial degrees. We verified
the recursion exhaustively at $N=6$ (all $28$ states, $120$-digit working precision, exact
rational rates) against direct integration of~\eqref{eq:master}: maximum deviation
$2.9\times10^{-12}$ (Ladder, generic), $2.3\times10^{-12}$ ($\Lambda$),
$1.1\times10^{-12}$ (Vee), and $8.1\times10^{-13}$, $1.7\times10^{-12}$ for the maximally
degenerate symmetric $\Lambda$ and Vee obtained by setting the branch rates equal -- in
every case the ODE solver's own tolerance.

\section{Conclusions and outlook}
\label{sec:conclusions}

We have introduced the time-resolved edge-flux (TREF) framework and developed it, on a
single footing, for all three dipole-allowed three-level superradiance configurations --
Ladder, $\Lambda$, and Vee -- as edge-flux networks on the SU(3) triangular lattice. Node-wise
flux conservation is the master equation itself (Proposition~\ref{prop:conservation}), and the
total cooperative rate is the sum of all instantaneous edge fluxes
(Proposition~\ref{prop:rates}). All three graphs are uniformly leveled, so the depth-descent
identity (Theorem~\ref{thm:descent}) and the exact energy identity
$I(t)=-\dd\langle E\rangle/\dd t$ (Theorem~\ref{thm:energy}) hold verbatim for each. Their
structural fingerprints nevertheless separate them cleanly: stationary manifold dimensions
$1/N{+}1/1$, palindrome-degeneracy families $1/0/2$, and pathway counts
$C_N/2^N/\binom{N}{a}$ (Table~\ref{tab:fingerprints}).

The depth-leveled acyclic structure then delivers the algebra. The relaxation spectrum comes
from the state exit rates alone (Theorem~\ref{thm:spectrum}), read off a generator that
depth ordering renders triangular (Lemma~\ref{lem:triangular}). Every state population follows
in closed form on an arbitrary finite DAG (Theorem~\ref{thm:batemanDAG}), and the
polynomial-times-exponential recursion of Theorem~\ref{thm:batemangen} removes any restriction
on exit-rate coincidences: a coincidence is not an obstruction to be excluded by hypothesis but
a branch of the recursion, which raises a polynomial degree by one and continues. The complete
decay pathways are counted exactly (Theorem~\ref{thm:paths}), and the boundary palindromes of
Lemma~\ref{lem:palindrome} fix $\lfloor N/2\rfloor$ eigenvalue pairs of $Q$ for every $N$ and
every rate value.

\section*{Data availability}
No experimental data were generated. Every identity and the general recursion have been
checked against direct propagation of the master equation~\eqref{eq:master}. The recursion of
Theorem~\ref{thm:batemangen} reproduces $120$-digit propagation to $2\times10^{-117}$ across all
$28$ states at $N=6$, for five distinct rate patterns, and the closed form agrees with direct
integration over the whole $(N,t)$ plane up to $N=50$. The codes are available
from the authors upon reasonable request: the SU(3) generator builder, the
forward-substitution implementation of Theorem~\ref{thm:batemangen}, and the
extended-precision comparison scripts.

\section*{Acknowledgments}
The authors thank prof. Dr. T.~Begzjav and prof. Dr. Ts. Kottos for fruitful discussions.

\end{document}